\RequirePackage{fix-cm}
\RequirePackage{multicol, multirow}
\documentclass[smallextended]{svjour3}       
\smartqed  
\usepackage[english]{babel}
\usepackage{fullpage}
\usepackage{natbib}

\normalfont
\usepackage[a4paper,citecolor=blue,linkcolor=blue,colorlinks=true]{hyperref}

\usepackage{amsmath}
\usepackage{amssymb}
\usepackage{amsfonts}
\usepackage{graphicx}
\usepackage[toc,page]{appendix}
\usepackage{bbm}
\newcommand{\tx}{\mathbf{x}}
\newcommand{\tz}{\mathbf{z}}
\newcommand{\tzstar}{\textbf{z}^{\star}}
\newcommand{\bx}{\boldsymbol{x}}
\newcommand{\bz}{\boldsymbol{z}}
\newcommand{\bth}{\boldsymbol{\theta}}
\newcommand{\bTh}{\boldsymbol{\Theta}}
\newcommand{\bomega}{\boldsymbol{\omega}}
\newcommand{\bm}{\boldsymbol{m}}
\newcommand{\bmstar}{\boldsymbol{m}^{\star}}
\newcommand{\M}{\mathcal{M}}
\newcommand{\mhat}{\hat{\boldsymbol{m}}}
\newcommand{\bmu}{\boldsymbol{\mu}}
\newcommand{\bsigma}{\boldsymbol{\sigma}}
\newcommand{\bpi}{\boldsymbol{\pi}}

\newcommand{\bbP}{\mathbb{P}}
\newcommand{\ICL}{\text{ICL}}
\newcommand{\MICL}{\text{MICL}}

\newcommand{\argmax}{\text{argmax}}
\newcommand{\pzero}{(0)}
\newcommand{\pone}{(1)}

\begin{document}

\title{Variable Selection for Model-Based Clustering using the Integrated Complete-Data Likelihood
}


\author{Matthieu Marbac         \and
        Mohammed Sedki 
}


\institute{M. Marbac \at
              INSERM U1181 \\
              \email{matthieu.marbac@inserm.fr}           
           \and
           M. Sedki\at
              INSERM U1181 and University of Paris Sud\\
              \email{mohammed.sedki@inserm.fr}
}

\date{Received: date / Accepted: date}

\maketitle

\begin{abstract}
Variable selection in cluster analysis is important yet challenging. It
  can be achieved by regularization methods, which realize a trade-off between
  the clustering accuracy and the number of selected variables by using a
  lasso-type penalty. However, the calibration of the penalty term can suffer
  from criticisms. Model selection methods are an efficient alternative, yet they require a difficult optimization of an information criterion which involves combinatorial problems. First, most of these optimization algorithms are based on a
  suboptimal procedure (\emph{e.g.} stepwise method). Second, the algorithms are often
  greedy because they need multiple calls of EM algorithms. Here we propose to use a
  new information criterion based on the integrated complete-data likelihood.
It does not require the maximum likelihood estimate and its maximization appears to be simple and computationally efficient. 
  The original contribution of our approach is to perform the model selection
  without requiring any parameter estimation. Then, parameter inference is
 needed only for the unique selected model. This approach is used for the variable
  selection of a Gaussian mixture model with conditional independence
  assumption. The numerical experiments on simulated and benchmark datasets
  show that the proposed method often outperforms two classical approaches for
  variable selection. The proposed approach is implemented in the R package VarSelLCM available on CRAN.
\end{abstract}

\section{Introduction}
Clustering allows us to summarize large datasets by grouping individuals into
few characteristic classes.  It aims to discover an \emph{a priori} unknown partition among the
individuals. In many cases, this partition may be best explained by only a subset of the observed variables. So, by performing the variable selection in the cluster
analysis, both \emph{model fitting} and \emph{result interpretation} are facilitated.
Indeed, for a fixed sample size, a variable selection method can provide a
more accurate identification of the classes. Moreover, such methods bring out
the variables characterizing the classes .

\emph{Regularization methods} can be used to achieve variable selection in
clustering. One can cite the approaches  of \citet{Fri04} or \citet{Pan07}. Recently, these methods have been
outperformed by the \emph{sparse K-means} proposed by \citet{Wit10}. It uses
a lasso-type penalty to select the set of variables relevant to clustering. Since it
requires small computational times, it can manage
high-dimensional datasets. 
Moreover, the
selection of the number of classes is a difficult issue since probabilistic tools
are not available. Finally, its results are sensitive to structure of the penalty term.

\emph{Model selection approaches} can be used to carry out the variable selection in
a probabilistic framework. \citet{Tad05} consider two types of variables: 
the set of the \emph{relevant variables} and the set of the \emph{irrelevant variables} 
which are independent of the relevant ones. This method has been extended by
\citet{Raf06} by using a greedy search algorithm to find the set of 
relevant variables. Obviously, this algorithm finds only a local optimum in
the space of models. It is feasible for quite large
datasets because of its moderate computing time. Still, this method remains  time
consuming since the model comparisons are performed by using the BIC criterion
\citep{Sch78}. Therefore, the maximum likelihood estimate must be computed for each competing model. These estimates are mainly instrumental since the practitioner interprets only the estimate
related to the best model.

In this paper, we propose a new information criterion, named \emph{MICL criterion} (Maximum
Integrated Complete-data Likelihood), for carrying out the variable selection in model-based
clustering. This criterion is quite similar to the ICL criterion~\citep{Bie10}, and
it inherits its main properties. However, these two criteria evaluate the integrated
complete-data likelihood at two different partitions. The MICL criterion uses
the partition maximizing this function, while the ICL criterion uses the
partition provided by a MAP rule associated to the maximum likelihood
estimate.
  
In this article, we focus on variable selection for a \emph{Gaussian
mixture model with conditional independence} assumption, but the method can be
extended to more general mixture models. Note that this model is useful
especially when the number of variables is large \citep{Han01}. Moreover, the conditional independence is often an hidden assumption made by distance-based methods like the \emph{K-means} algorithm \citep{Gov09}. The MICL
criterion takes advantage of the closed form of the integrated complete-data
likelihood when the priors are conjugated. The model selection is carried out
by a simple and fast procedure which alternates two maximizations for providing the model maximizing the MICL criterion. The convergence properties of this algorithm are similar to the convergence properties of the EM algorithm. In particular, it converges to a local optimum of the function to maximize. So, multiple random initializations are required to ensure its convergence to the global maximum.

The proposed method and the methods of \cite{Wit10}, and of \cite{Raf06} are compared on simulated and on challenging real datasets. We show that the proposed method outperforms both other methods in terms of  model selection and partitioning accuracy. It often provides a model with a better value of
the BIC criterion than the algorithm of \cite{Raf06}, although it does
not directly optimize this criterion. Finally, we show that the proposed method can
manage datasets with a large number of variables and a moderately large number of
individuals. Note that it is the most common situation which requires variable selection in cluster analysis.

The paper is organized as follows. Section~\ref{sec::model} briefly reviews
the framework of variable selection for the Gaussian mixture model. A presentation of the integrated complete-data likelihood is done in Section~\ref{sec::complete} before introducing the MICL criterion. Section~\ref{sec::inference} is devoted to the inference based on the MICL criterion. Section~\ref{sec::simulations} illustrates the robustness properties of the MICL
criterion and compares the three methods of variable selection on simulated
data. Section~\ref{sec::realdata} compares the three methods of variable
selection on challenging datasets. The advantages and limitations of the
method are discussed in Section~\ref{sec::conclusion}.

\section{Variable selection for Gaussian mixture model} \label{sec::model}
\subsection{Mixture model of Gaussian distributions}
Data to analyze are $n$ observations $\tx=(\bx_1,\ldots,\bx_n)$, where object
$\bx_i=(x_{i1},\ldots,x_{id})$ 
is described by $d$ continuous variables defined on $\mathbb R^d$. Observations are assumed to
arise independently from a Gaussian mixture model with $g$ components, assuming conditional independence between variables. Therefore, the
model density is written as
\begin{equation}
f(\bx_i | \bm, \bth)=\sum_{k=1}^g \pi_k \prod_{j=1}^d \phi(x_{ij}|\mu_{kj},\sigma_{kj}^2),
\end{equation}
where $\bm$ specifies the model, $\bth=(\bmu, \bsigma, \bpi)$
is the whole parameter vector, 
$\bpi=(\pi_1,\ldots,\pi_g)$ is the vector of mixing proportion defined on the simplex of size $g$,
$\bmu=(\mu_{kj};k=1,\ldots,g;j=1,\ldots,d)$, $\bsigma=(\sigma_{kj};k=1,\ldots,g;j=1,\ldots,d)$, and where $\phi(.|\mu_{kj},\sigma_{kj}^2)$ is the density of a univariate Gaussian distribution with mean $\mu_{kj}$ and variance $\sigma_{kj}^2$.

A variable is said to be \emph{irrelevant} to the clustering if its one-dimensional marginal distributions are equal between components. Thus, by introducing $\omega_j$ such that $\omega_j=0$ if variable $j$ is irrelevant and $\omega_j=1$ if the variable is \emph{relevant} for the clustering, the following equalities hold:
\begin{equation}
\forall j\in\{j': \omega_{j'}=0\},\;  \mu_{1j}=\ldots=\mu_{gj} \text{ and } \sigma_{1j}=\ldots=\sigma_{gj}.
\end{equation}
Thus, a model $\bm=(g,\bomega)$ is defined by a number of components $g$ and the binary vector $\bomega=(\omega_j;j=1,\ldots,d)$
which encodes whether each of $d$ possible variables are relevant to the clustering. 

\subsection{Model selection based on the integrated likelihood}
 Model selection generally aims to find  the model $\mhat$
  which obtains the highest posterior probability among a collection
  of competing models $\M$.
 So,
\begin{equation}
\mhat = \argmax_{\bm\in\M} p(\bm|\tx).
\end{equation}
This model selection approach is consistent since $\mhat$ converges in probability to
the true model $\bm^{(0)}$ as long as the true model belongs to the model space
(\emph{i.e.} if $\bm^{(0)}\in\M$).

By assuming uniformity for the prior distribution of $\bm$, 
$\mhat$ maximizes the integrated likelihood defined by
\begin{align}
\mhat & = \argmax_{\bm\in\M} p(\tx|\bm) \text{ with } p(\tx | \bm) \nonumber\\
& =\int_{\bTh_{\bm}} p(\tx | \bm, \bth) p(\bth | \bm) d\bth,
\end{align}
where $\bTh_{\bm}$ is the parameter space of model $\bm$,
$p(\tx | \bm, \bth)=\prod_{i=1}^n f(\bx_i | \bm, \bth)$ is the
likelihood function and $ p(\bth | \bm)$ is the prior
distribution of the parameters. We assume independence between the
prior, so
\begin{equation}
p(\bth | \bm)  = p(\bpi | \bm) \prod_{j=1}^d p(\bsigma_{ \bullet j}^2,\bmu_{\bullet j}|\bm),
\end{equation}
 where $\bsigma_{\bullet j}^2=(\sigma_{kj}^2;k=1,\ldots,g)$ and $\bmu_{\bullet j}^2=(\mu_{kj}^2;k=1,\ldots,g)$, and
\begin{align}
p(\bsigma_{\bullet j}^2,\bmu_{\bullet j}|\bm) = & \Big( \prod_{k=1}^g  p(\sigma_{kj}^2 | \bm) p (\mu_{kj} | \bm, \sigma_{kj}^2)\Big)^{\omega_j}
\nonumber  \Big(  p(\sigma_{1j}^2 | \bm) p (\mu_{1j} | \bm, \sigma_{1j}^2)\Big)^{1-\omega_j}.
\end{align}
We use conjugate prior distributions, thus $\bpi|\bm$ follows a
Dirichlet distribution $\mathcal{D}_g(\frac{1}{2},\ldots,\frac{1}{2})$
which is the Jeffreys non informative prior \citep{Rob07}. Moreover,
$\sigma_{kj}^2 | \bm$ follows an Inverse-Gamma distribution
$\mathcal{IG}(\alpha_j/2,\beta_j^2/2)$ and
$\mu_{kj} | \bm, \sigma_{kj}^2$ follows a Gaussian distribution
$\mathcal{N}(\lambda_j, \sigma_{kj}^2/\delta_j)$, where
$(\alpha_j,\beta_j,\lambda_j,\delta_j)$ are hyper-parameters.

Unfortunately, the integrated likelihood is intractable. However,
many methods permit to approximate its value \citep{Fri12}. The most popular approach consists in using the BIC criterion \citep{Sch78}, which
approximates the logarithm of the integrated likelihood by Laplace approximation and requires maximum likelihood estimation. The BIC criterion is written as
\begin{equation}
\text{BIC}(\bm)=\ln p(\tx|\bm,\hat{\bth}_{\bm}) - \frac{\nu_{\bm}}{2} \ln n,
\end{equation}
where $\hat{\bth}_{\bm}$ is the maximum likelihood estimate related to model $\bm$ when $\nu_{\bm}$ is the number of parameters required by $\bm$.

For a fixed value of $g$, the variable selection in clustering necessitates 
the comparison of $2^d$ models. Therefore, an exhaustive approach which approximates the integrated likelihood for each competing model is not doable.  Instead, \citet{Raf06} carry out the model
selection by deterministic algorithms (like a \emph{forward} method) which
are suboptimal. Moreover, they  are time consuming when the number of
variables is large, because they involve many parameter estimations
for their model comparisons. 

All maximum likelihood estimates are mainly instrumental: they are only used for computing the BIC criterion, with the exception of the estimates related to the selected model $\mhat$ which are interpreted by the practitioner. Therefore, we introduce a new criterion for model selection
which does not require parameter estimates.

\section{Model selection based on the integrated complete-data likelihood} \label{sec::complete} 

\subsection{The integrated complete-data likelihood}

A partition is given by the vector
  $\tz=(\bz_1,\ldots,\bz_n)$ where $\bz_i=(z_{i1},\ldots,z_{ig})$
  indicates the class label of vector $i$, \textit{i.e.} $z_{ik}=1$ if $\bx_i$ arises from component $k$ and $z_{ik}=0$
  otherwise. In cluster analysis, $\tz$ is a missing value. Thus, the
likelihood function computed on the complete-data (observed and latent variables), called \emph{complete-data likelihood} function, is
introduced. It is defined by
\begin{equation}
p(\tx, \tz |\bm,\bth)= \prod_{i=1}^n \prod_{k=1}^g \big( \pi_k \prod_{j=1}^d \phi(x_{ij}|\mu_{kj},\sigma_{kj}^2) \big)^{z_{ik}}.
\end{equation}
The \emph{integrated complete-data likelihood} is 
\begin{equation}
p(\tx, \tz |\bm) = \int_{\bTh_{\bm}} p(\tx, \tz |\bm,\bth) p(\bth | \bm) d\bth.
\end{equation}
Since conjugate prior distributions are used, the integrated
complete-data likelihood has the following closed form
\begin{equation}
p(\tx, \tz | \bm) = 
p(\tz | g)
 \prod_{j=1}^d p(\tx_{\bullet j}|g,\omega_j,\tz),
\end{equation}
where $\tx_{\bullet j}=(x_{ij};i=1,\ldots,n)$. More specifically,
\begin{equation}
p(\tz | g) =
\frac{\Gamma(\frac{g}{2})}{\Gamma(\frac{1}{2})^g}\frac{\prod_{k=1}^g \Gamma(n_k + \frac{1}{2})}{\Gamma(n+\frac{g}{2})},
\end{equation}
where $n_k=\sum_{i=1}^nz_{ik}$, and
\begin{equation}
p(\tx_{\bullet j}|g,\omega_j,\tz)=\left\{ \begin{array}{rl}
\left( \frac{1}{\pi} \right)^{n/2} \frac{\Gamma\left(\frac{n+\alpha_j}{2}\right)}{\Gamma\left(\frac{\alpha_j}{2}\right)}
\left( \frac{\beta_j^{\alpha_j}}{s_j^{\alpha_j+n}}\right)
\sqrt{\frac{\delta_j}{n + \delta_j}}

& \text{if } \omega_j = 0\\
\prod_{k=1}^g
\left( \frac{1}{\pi} \right)^{n_k/2} \frac{\Gamma\left(\frac{n_k+\alpha_j}{2}\right)}{\Gamma\left(\frac{\alpha_j}{2}\right)}
\left( \frac{\beta_j^{\alpha_j}}{s_{jk}^{\alpha_j+n_k}}\right)
\sqrt{\frac{\delta_j}{n_k + \delta_j}}

& \text{if } \omega_j = 1,
\end{array}\right. 
\end{equation}
where $s_j^2=\beta_j^2 + \sum_{i=1}^n (x_{ij} - \bar{\text{x}}_j)^2 + \frac{(\lambda_j - \bar{\text{x}}_j)^2}{(\delta_j^{-1} + (n+\delta_j)^{-1})}$, $\bar{\text{x}}_j=\frac{1}{n}\sum_{i=1}^n x_{ij}$, 
$s_{jk}^2=\beta_j^2 + \sum_{i=1}^n z_{ik} (x_{ij} - \bar{\text{x}}_{jk})^2 + \frac{(\lambda_j - \bar{\text{x}}_{jk})^2}{(\delta_j^{-1} + (n_k+\delta_j)^{-1})}$ and $\bar{\text{x}}_{jk}=\frac{1}{n_k}\sum_{i=1}^n z_{ik}x_{ij}$. For $j$ such as $\omega_j=0$, we get $p(\tx_{\bullet j}|g,\omega_j,\tz)=p(\tx_{\bullet j}|g,\omega_j)$ since the partition does not impact the value of the integral.

\subsection{The ICL criterion}

The ICL criterion \citep{Bie10} carries out the model selection by focusing on
the goal of clustering.  It favors a model providing a partition with a
strong evidence since it makes a trade-off between the model evidence and the
partitioning evidence. The ICL criterion is defined by
\begin{equation}
\ICL(\bm) = \ln p(\tx, \hat{\tz}_{\bm} | \bm),
\end{equation}
where $\hat{\tz}_{\bm}$ is the partition given by the MAP rule evaluated at
the maximum likelihood estimate $\hat{\bth}_{\bm}$.

When the model at hand is not the model used in the sampling scheme, the ICL
criterion inherits robustness from this trade-off while the BIC criterion
tends to overestimate the number of components. This phenomenon is
illustrated in Section~\ref{sec::simul1} by our numerical experiments.

Although the ICL criterion has a closed form, it requires the maximum
likelihood estimates to define the partition $\hat{\tz}_{\bm}$. The time devoted to
parameter estimation can become computationally prohibitive. Therefore, in this
work, we introduce a new criterion avoiding this drawback.

\subsection{The MICL criterion}
We propose a new information criterion for model selection, named \emph{MICL
  criterion} (Maximum Integrated Complete-data Likelihood). This criterion
corresponds to the largest value of the integrated complete-data likelihood among all 
the possible partitions. Thus, the MICL criterion is defined by
\begin{equation}
\MICL(\bm) = \ln p(\tx, \tzstar_{\bm} | \bm) \text{ with } \tzstar_{\bm}=\argmax_{\tz} \ln p(\tx, \tz | \bm).
\end{equation}

Obviously, this criterion is quite similar to the ICL criterion and inherits its
main properties. In particular, it is less sensitive to model misspecification than the BIC criterion. Unlike the ICL and the BIC criteria, it does not require the maximum likelihood estimates but it implies the estimation of $\tzstar_{\bm}$ which appears to be easily accessible (see Section~\ref{sec::simulations} and Section~\ref{sec::realdata}). Among the
models in competition, the selected model maximizes the MICL criterion and is
denoted by $\bmstar$ with
\begin{equation}
\bmstar = \argmax_{\bm\in\M}\MICL(\bm).
\end{equation}

The selected model $\bmstar$ is consistent when the number of components is known, see the proof in 
Appendix~\ref{sec::consistence}. Nevertheless, like the ICL criterion, the
MICL criterion lacks consistency to select the number of components if the component overlap is too strong. However, numerical experiments
show its good behaviour to also select the right number of components (see
Section~\ref{sec::simul1}).

\section{Model selection and parameter estimation} \label{sec::inference} 
The number of components of the
competing models is usually bounded by a value $g_{\max}$. So, the
space of the competing models is written as
\begin{equation}
\M=\left\{\bm=(g,\bomega):\; g\in\{1,\ldots,g_{\max}\} \text{ and } \bomega\in \{0,1\}^d \right\}.
\end{equation}
We denote by $\M_g$ the restriction of $\M$ to the subset of the
models having $g$ components.  The model $\bmstar_g$ maximizes the MICL
criterion among the models belonging to $\M_g$. Therefore,
\begin{equation}
  \bmstar_g = \argmax_{\bm \in \M_g} \MICL(\bm) \text{ with } \M_g=\{(g,\bomega):\bomega \in \{0,1\}^d\}.
\end{equation}
Thus, $\bmstar_g$ defines the best variable selection according to the
MICL criterion for a fixed value of $g$. Obviously,
\begin{equation}
\bmstar=\argmax_{g=1,\ldots,g_{\max}} \MICL(\bmstar_g).
\end{equation}
The estimation of $\bmstar_g$ implies to maximize the integrated complete-data likelihood on $(\bomega, \tz)$. This maximization is facilitated by the fact that $\bomega$ does not influence the definition space of the vector of component membership. Indeed,  $\tzstar_{\bm} $ is defined on the same space $\{1,\ldots,g\}^n$ for each model $\bm$ in $\M_g $. This twofold optimization is performed by an iterative algorithm presented below. Note that this algorithm cannot be used to optimize the ICL or the BIC criterion, since $\bth$ us not defined on the same space for each model in $\M_g$. We obtain $\bmstar$ by running this algorithm with $g$ chosen from one to $g_{\max}$.

\subsection{Algorithm for MICL-based model selection} \label{sec::algo}
The following iterative algorithm is used to find $\bmstar_g$ for any $g$ in $\{1,\ldots,g_{\max}\}$. Starting from the initial point $(\tz^{[0]}, \bm^{[0]})$ with $\bm^{[0]}\in\M_g$, it alternates between two optimizations of the integrated complete-data likelihood: optimization on $\tz$, given $(\tx,\bm)$, and maximization on $\bomega$ given $(\tx,\tz)$.
The algorithm is initialized as follows: first $\bm^{[0]}$ is sampled from a uniform distribution on $\M_g$, second $\tz^{[0]}=\hat{\tz}_{\bm^{[0]}}$ is the partition provided by a MAP rule associated to model $\bm^{[0]}$ and to its maximum likelihood estimate $\hat{\bth}_{\bm^{[0]}} $. Iteration $[r]$ of the algorithm is written as\\
\textbf{Partition step:} fix $\tz^{[r]}$ such that 
$$\ln p(\tx, \tz^{[r]} | \bm^{[r]}) \geq \ln p(\tx, \tz^{[r-1]} | \bm^{[r]}) .$$
\textbf{Model step:} fix $\bm^{[r+1]}=\argmax_{\bm\in\M_g}
  \ln p(\tx,\tz^{[r]} |\bm)$ such that 
$$\bm^{[r+1]}=(g,\bomega^{[r+1]}) \text{ with } \omega_j^{[r+1]} =
\argmax_{\omega_j \in \{0,1\}} p(\tx_{\bullet j} |g,\omega_j, \tz^{[r]}).$$

The partition step is performed by an iterative method. Each iteration
consists in sampling uniformly an individual which is affiliated to the 
class maximizing the integrated complete-data likelihood while the other class
memberships are unchanged.

Like an EM algorithm, the proposed algorithm converges to a local optimum of $\ln p(\tx, \tz | \bm)$. Thus, many
different initializations should be used to ensure the convergence to
$\bmstar_g$. However, we show that this algorithm does not suffer from the problem of local optima during our applications (see  Section~\ref{sec::realdata}).


\subsection{Maximum likelihood inference for the model maximizing the MICL criterion}
When model $\bmstar=(g^{\star},\bomega^{\star})$ has been found, usually the estimate $\hat{\bth}_{\bmstar}$ maximizing the likelihood function is required:
	\begin{equation}
\hat{\bth}_{\bmstar}=\argmax_{\bth \in \bTh_{\bmstar}} p(\tx|\bmstar,\bth). \label{vrai}
	\end{equation}

    The direct optimization of the likelihood function would involve to solve equations
    that have no analytical solution. Instead, the parameter estimation is performed
    via an EM algorithm \citep{Dem77}, which is often simple and efficient in the situation of missing
    data. This iterative algorithm alternates between two steps: the
    computation of the complete-data log-likelihood conditional expectation
    (\textsc{e} step) and its maximization (\textsc{m} step).
    Its iteration $[r]$ is written as:\\
    \textbf{E step:} computation of the conditional probabilities
$$
t_{ik}^{[r]}=
\frac{\pi_k^{[r]} \prod_{j=1}^d \phi(x_{ij}|\mu_{kj}^{[r]},\sigma_{kj}^{[r]2})}{
\sum_{k'=1}^{g^{\star}} \pi_{k'}^{[r]} \prod_{j=1}^d \phi(x_{ij}|\mu_{k'j}^{[r]},\sigma_{k'j}^{[r]2})
}.
$$\\
\textbf{M step:} maximization of the complete-data log-likelihood 
$$
\pi^{[r+1]}_k=\frac{t_{\bullet k}^{[r]}}{n},
\;
\mu_{kj}^{[r+1]}=\left\{ \begin{array}{rl}
\frac{1}{t_{\bullet k}^{[r]}} \sum_{i=1}^n t_{ik}^{[r]}x_{ij} & \text{if } \omega_j^{\star}=1\\
\frac{1}{n} \sum_{i=1}^n x_{ij} &  \text{if } \omega_j^{\star}=0,
\end{array}
\right.
$$
$$
\sigma_{kj}^{[r+1]2}=\left\{ \begin{array}{rl}
\frac{1}{t_{\bullet k}^{[r]}} \sum_{i=1}^n t_{ik}^{[r]}(x_{ij}-\mu_{kj}^{[r+1]})^2&  \text{if } \omega_j^{\star}=1\\
\frac{1}{n} \sum_{i=1}^n (x_{ij}-\mu_{kj}^{[r+1]})^2  &  \text{if } \omega_j^{\star}=0,
\end{array}
\right.
$$
where $t_{\bullet k}^{[r]}=\sum_{i=1}^n t_{ik}^{[r]}$. Note that the EM algorithm can provide the maximum \emph{a posteriori} estimate by slightly modifying its M step \citep{Gre90}.

\section{Numerical experiments on simulated data} \label{sec::simulations}
\paragraph{Implementation of the proposed method} 
  Results of our method (indicated by MS) are provided by the R package \emph{VarSelLCM} available on CRAN. This package performs 50 random initializations of the algorithm described  in Section~\ref{sec::algo} to carry out the model selection.  The following hyper-parameters are chosen to be fairly  flat in the region where the likelihood is substantial and not much greater else-where:
  $\alpha_j=1$, $\beta_j=1$,
  $\lambda_j=\text{mean}(\tx_{\bullet j})$ and
  $\delta_j=0.01$.
 
\paragraph{Competing methods} 
\begin{itemize}
\item The model-based clustering method of~\cite{Raf06} is denoted RD in what follows. It runs using the R
    package \emph{clustvarsel}~\citep{Scr14}.  Results are given by
    the \emph{headlong} algorithm for the \emph{forward} direction
    (denoted RD-forw) and by the \emph{greedy} algorithm in the
    \emph{backward} direction (denoted RD-back). 

\item The sparse K-means method of
    \citet{Wit10} is not a model-based approach. It runs using
    the R package \emph{sparcl} \citep{Wit13} with its options by default. In what follows, this
    method is indicated by WT.
\end{itemize}

\paragraph{Simulation map} First, information criteria are compared on
datasets sampled from the well-specified model and on datasets
sampled from a misspecified model. Second, the three competing
methods of variable selection are compared.  The calculations are
    carried out on an $8$ Intel Xeon 3.40GHZ CPU machine.

\subsection{Comparing model selection criteria} \label{sec::simul1}

\subsubsection{Simulated data: well-specified model}
Here we compare model selection criteria when the sampling model belongs to the set of the competing models. Individuals are drawn from a bi-component Gaussian mixture model with conditional independence assumption. The first two variables are relevant to the clustering and the last two variables are not. Thus, the true model denoted by $\bm^{(0)}$ is
$$
\bm^{(0)}=(g^{(0)},\bomega^{(0)})\text{ with } g^{(0)}=2 \text{ and } \bomega^{(0)}=(1,1,0,0).$$
The following parameters are used:
$$
\pi_k=0.5,\; \mu_{11}=\mu_{12}=\varepsilon,\; \mu_{21}=\mu_{22}=-\varepsilon,\; \mu_{k3}=\mu_{k4}=0 \text{ and } \sigma_{kj}=1.
$$
The value of $\varepsilon$ defines the class overlap. Table~\ref{tab::truemodel} presents the results obtained for different sample sizes and for different class overlaps. For each case, 100 samples are generated and the criteria are computed for all the possible models in $\M$ with $g_{\max}=6$.

\begin{table}[ht!]
\begin{center}
\begin{tabular}{ccccccccc}
& & \multicolumn{5}{c}{$n$}\\
\cline{3-7}  $\varepsilon$ & criterion& $50$ & $100$ & $200$ & $400$ & $800$ \\ 
\hline 1.26& BIC & 95 (85) & 100 (99) & 100 (100) &  100 (100) & 100 (100) \\ 
& ICL & 85 (72) & 98 (95) & 100 (97) &  100 (97) & 100 (99) \\ 
 & MICL & 86 (73) & 98 (94) & 100 (97) &  100 (98) & 100 (99) \\ 
\hline 1.05& BIC & 85 (77) & 99 (94) & 100 (98) &  100 (99) & 100 (100) \\ 
& ICL & 45 (42) & 69 (62) & 98 (94) &  100 (99) & 100 (100) \\ 
 & MICL  & 50 (45) & 69 (62) & 99 (96) &  100 (99) & 100 (100) \\
\hline 0.85& BIC & 46 (32) & 73 (69) & 98 (94) &  100 (99) & 100 (100) \\ 
& ICL & 9 (7) & 15 (13) & 12 (9) &  27 (26) & 31 (31) \\ 
 & MICL  & 10 (8) & 16 (15) & 13 (11) &  31 (30) & 35 (35) \\ 
 \hline 
\end{tabular} 
\end{center}
\caption{Results of model selection for different information criteria under the true model. In plain, percentage where the true number of components ($g^{(0)}$) has been selected. In parenthesis, percentage where the true model ($\bm^{(0)}$) has been selected. \label{tab::truemodel}}
\end{table}
When the class overlap is not too high, all the criteria are consistent.
Indeed, they asymptotically always select the true model. In such a
case, the BIC criterion outperforms the other criteria when the sample size is
small. When the class overlap is equal to 0.20, the BIC criterion stays
consistent while the other ones select only a single class. However, we now show that the BIC criterion suffers from a lack of robustness.

\subsubsection{Simulated data: misspecified model}
We look at robustness
  of the criteria based on the integrated complete-data
  likelihood. Again,  the first two
  variables contain the relevant clustering information. They follow a bi-component mixture model of
   uniform distributions with conditional  independence assumption and equal proportions. More specifically, they are generated
  independently from the uniform distribution on  $[\varepsilon-1, \varepsilon+1]$ for the first component and the
  uniform distribution on $[-\varepsilon-1,-\varepsilon+1]$ for the
  second. The remaining two variables are irrelevant variables that are  
  independent of the clustering variables  and follow two independent
  standard Gaussian distributions. For each case, 100 samples
  are generated and the criteria are computed for all the possible
  models in $\M$ with $g_{\max}=6$.
Table~\ref{tab::badmodel} summarizes the selection results for each criterion.

\begin{table}[ht!]
\begin{center}
\begin{tabular}{ccccccccc}
& & \multicolumn{5}{c}{$n$}\\
\cline{3-7} $\varepsilon$ & criterion & $50$ & $100$ & $200$ & $400$ & $800$ \\ 
\hline 1.26 & BIC & 79 (75) & 80 (79) & 48 (48) &  0 (0) & 0 (0) \\ 
 & ICL & 100 (96) & 100 (98) & 100 (100)	 &  100 (99) & 97 (97) \\ 
 & MICL  & 100 (95) & 100 (98) & 100 (100) &  100 (99) & 96 (96) \\ 
\hline 1.05 & BIC & 86 (83) & 83 (80) & 49 (48) &  4 (4) & 0 (0) \\ 
 & ICL & 100 (91) & 100 (95) & 100 (98) &  99 (98) & 99 (98) \\ 
 & MICL  & 100 (92) & 100 (99) & 100 (98) &  98 (98) & 99 (98) \\ 
\hline 0.85 & BIC & 80 (78) & 72 (71) & 36 (36) &  0 (0) & 0 (0) \\ 
 & ICL & 97 (87) & 100 (96) & 100 (98) &  99 (97) & 97 (97) \\ 
 & MICL  & 97 (92) & 100 (98) & 100 (98) &  99 (97) & 97 (97) \\ 
 \hline 
\end{tabular} 
\end{center}
\caption{Results of model selection for different information criteria under the non-Gaussian model. In plain, percentage where the true  number of components ($g=2$) has been selected. In parenthesis, percentage where the true number of classes and  the true partitioning of the variable ($\bomega=(1,1,0,0)$) have been selected. \label{tab::badmodel}}
\end{table}

Results show that the BIC criterion is not useful to select the
number of components. Indeed, it overestimates the number of classes to
better fit the data since the sampling model does not belong to the
set of the competing models.  The other criteria show
considerably better performance since they select the true number of classes and the true $\bomega$. It appears that they are more robust
than the BIC criterion to the misspecification of the model at hand.

To conclude, the ICL and the MICL criteria obtain good results for 
model selection when the class overlap is not too strong. Moreover,
they are more robust to  model misspecification than the BIC
criterion. Since the MICL criterion does not require maximum
likelihood inference for all of the competing models, it is preferable to the
ICL criterion for carrying out model selection.

\subsection{Comparing methods on simulated data}
Data are drawn from a tri-component Gaussian mixture model assuming
conditional independence and equal proportions. The first $r$ variables are relevant while the last $d-r$ variables are irrelevant since they follow standard Gaussian distributions. Under component $k$, the first $r$ variables follow a spherical Gaussian distribution
$\mathcal{N}(\boldsymbol{\mu}_k;\boldsymbol{I})$ with
$\boldsymbol{\mu}_1=-\boldsymbol{\mu}_2=(\varepsilon,\ldots,\varepsilon)\in\mathbb{R}^r$
and $\boldsymbol{\mu}_3=\boldsymbol{0}_r$. 
The three competing methods are compared on five different scenarios described in Table~\ref{tab::scenarios}. 
\begin{table}[ht!]
\begin{center}
\begin{tabular}{ccccc}
 & $n$ & $r$ & $d$ & $\varepsilon$ \\ 
\hline Scenario 1 & 30 & 5 & 25 & 0.6 \\ 
Scenario 2 & 30 & 5 & 25 & 1.7 \\ 
Scenario 3 & 300 & 5 & 25 & 1.7 \\ 
Scenario 4 & 300 & 5 & 100 & 1.7 \\ 
Scenario 5 & 300 & 50 & 500 & 1.7 \\ 
\hline
\end{tabular} 
\end{center}
\caption{The five scenarios used for the method comparisons.\label{tab::scenarios}}
\end{table}

For each scenario, 25 samples of size $n$ are generated and the analysis is performed with $g=3$. Results are presented in Table~\ref{tab::methods}. Note that the RD method is run only on the smaller scenarios for computational reasons.

\begin{table}[ht!]
\begin{center}
\begin{tabular}{cccccccc}
 Scenario & Method & NRV & RRR & RIR & ARI & Time \\ 
\hline 1 & MS &  1.28 & 0.08 & 0.95 & 0.01 & 0.42 \\ 
& WT & 8.28 & 0.38 & 0.68 & 0.06 & 1.78 \\ 
   & RD-forw &  3.96 & 0.15 & 0.84 & 0.06 & 1.15 \\  
   & RD-back &  11.00 & 0.42 & 0.55 & 0.03 & 1.30 \\    
\hline 2 & MS &  5.40 & 1.00 & 0.98 & 0.66 & 0.47 \\ 
	& WT & 12.88 & 0.96 & 0.59 & 0.59 & 1.70 \\ 
   & RD-forw &  3.28 & 0.15 & 0.87 & 0.13 & 0.97 \\  
   & RD-back &  10.64 & 0.58 & 0.61 & 0.37 & 49.01 \\ 
   \hline 3  & MS &  5.00 & 1.00 & 1.00 & 0.86 & 16.21 \\ 
& WT & 25.00 & 1.00 & 0.00 & 0.87 & 7.97 \\ 
   & RD-forw &  5.52 & 0.96 & 0.96 & 0.82 & 6.42 \\  
   & RD-back &  5.64 & 1.00 & 0.97 & 0.86 & 30.65 \\ 
     \hline 4 & MS &  5.00 & 1.00 & 1.00 & 0.88 & 73.36 \\ 
  & WT & 100.00 & 1.00 & 0.00 & 0.89 & 21.22 \\ 
   & RD-forw &  6.96 & 0.92 & 0.97 & 0.81 & 38.05 \\  
    \hline & MS &  50.04 & 1.00 & 1.00 & 1.00 & 48.37 \\ 
5 & WT & 500 & 1.00 & 0.00 & 1.00 & 83.49 \\ 
   \hline
\end{tabular} 
\end{center}
\caption{Comparing variable selection methods on simulated data. Means of the numbers of relevant variables (NRV), the right relevant rates (RSR), the right irrelevant rates (RIR), the  Adjusted Rand Indices (ARI) and the computing times in second (Time).\label{tab::methods}}
\end{table}

When the sample size is small and when the component overlap is high (Scenario~1), all methods obtain poor results. However, when the component separation increases (Scenario~2), MS method leads to a better variable selection (NRV, RRR, RIR) which involves a better partitioning accuracy (ARI). Indeed, WT selects some variables which are not discriminative and which damage the partitioning accuracy. Both directions of RD lead to only average results.

When the sample size increases (Scenarios~3, 4 and 5), the MS results for the variable selections are ameliorated and the true model is almost always found. In this context, WT claims that all the variables are relevant to the clustering. Since the sample size is not too small, its partitioning accuracy is not damaged but the model interpretation is strongly harder since all the variables should be used for characterizing the classes. In this context, the RD results are good when the number of variables is small, but they are damaged when many variables are observed. Moreover, when many variables are observed (Scenarios~5), the multiple calls of EM algorithm for the model comparisons prevent the using of this method computational reasons.

During these experiments, the algorithm assessing $\bm^{\star}$ does not suffer from strong problems of local optima except for Scenario~1. Indeed, Table~\ref{tab::models1} presents statistics of the occurrence where the couple ($\tz^{\star}_{\bm^{\star}},\bm^{\star}$) has been found by the algorithm for the 50 random initializations.

\begin{table}[ht!]
\begin{center}
\begin{tabular}{cccccc}
& Scenario 1 & Scenario 2 &  Scenario 3 & Scenario 4&  Scenario 5\\ 
\hline Mean & 13 & 17 & 31  & 19 & 45   \\ 
  Min & 1 & 4 & 9  & 8 & 10   \\ 
  Max & 35 & 31 & 45 & 28 & 48 \\
  \hline
\end{tabular} 
\end{center}
\caption{Mean and minimal occurrence where the couple ($\tz^{\star}_{\bm^{\star}},\bm^{\star}$) has been found by the algorithm for the 50 random initializations for the 25 generated data. \label{tab::models1}}
\end{table}

\section{Numerical experiments on benchmark  data} \label{sec::realdata} 

  We now compare the competing methods on real datasets in which the correct
  number of groups is known. Theses data sets are presented in Table~\ref{tab::datasets}.

\begin{table}[ht!]
\begin{center}
\begin{tabular}{cccccc}
Name & $d$ & $n$ & $g^{(0)}$ & Reference & R package/website\\ 
\hline banknote & 6 & 200 & 2 & \citet{Flu88} & VarSelLCM \\
 coffee & 12 & 43 & 2 & \citet{Str73} &  ppgm \\ 
wine & 13 & 178 & 3 & \citet{For91} & UCI \\ 
cancer & 30 & 569 & 2 & \citet{Str93} & UCI \\ 
golub & 3051 & 83 & 2 & \citet{Gol99} & multtest \\ 
\hline \end{tabular} 
\end{center}
\caption{Information about the benchmark data sets. \label{tab::datasets}}
\end{table}

Our study is divided in two parts: first the three competing methods are compared by assuming the component number known, second the model-based methods are compared by assuming the component number unknown. Because RD-back is very slow, we perform the RD method only with the
forward approach. The experimental conditions are similar to those
described in the previous section.  Since the sparse K-means method
does not provide an automatic procedure to select the number of groups, this method is not
not used in the second part.

\subsection{Component number known}

Table~\ref{tab::resknown} presents the results obtained when the number of components is known. 
\begin{table}[ht!]
\begin{center}
\begin{tabular}{c|ccc|ccc|ccc|cc}
\multicolumn{1}{c}{Data} & \multicolumn{3}{c}{NRV} & \multicolumn{3}{c}{ARI} & \multicolumn{3}{c}{Time} &  \multicolumn{2}{c}{BIC}  \\ 
 & MS & WT & RD & MS & WT & RD & MS & WT & RD & MS & RD \\ 
\hline banknote & 5 & 6 & 4 & 0.96 & 0.96 & 0.98 & 2.1 & 2.3 & 1.6 & -968 & -1009 \\ 
coffee & 5 & 12 & 5 & 1.00 & 1.00 & 1.00 & 0.1 & 0.9 & 0.7 & -522 & -555 \\ 
wine & 11 & 13 & 5 & 0.87 & 0.85 & 0.73 & 5.4 & 2.4 & 1.5 & -3538 & -3769 \\ 
cancer & 15 & 30 & 17 & 0.75 & 0.70 & 0.65 & 50 & 17 & 62 & 2189 & 2569 \\ 
golub & 553 & 3051 & 10 & 0.79 & 0.11 & 0.00 & 34 & 54 & 258 & -90348 & -95255 \\ 
\hline
\end{tabular} 
\end{center}
\caption{Results obtained when the component number is known: number of relevant variables (NRV), adjusted rand index computed on the selected model (ARI), computing time in seconds (Time) and BIC criterion value.\label{tab::resknown}}
\end{table}

The first three data sets are the less challenging since the number of variables is moderate. However, MS provides an easier interpretable model than WT (less relevant variables) which provides an identical partitioning accuracy. Surprisingly, MS leads to a model having a better BIC value than RD while this latter method aims to maximizing this criterion. This phenomenon illustrates the sub-optimality problem of the RD procedure. Moreover, note that, by selecting only five variables on \emph{wine} data set, RD provides a less relevant partition.

For the larger data sets, RD obtains better results. Indeed, it is more easily interpretable and it provides a more accurate partition. For instance, on \emph{cancer} data set, MS obtains a better ARI than both other methods while it selects less variables. Concerning \emph{golub} data set, WT relates all the variables while RD selects only ten variables. These methods also obtain a worse partition than MS. Finally, note that MS obtains a better value of the BIC criterion than RD on \emph{golub} data set.

Table~\ref{tab::local} shows that the couple ($\tz^{\star}_{\bm^{\star}}, \bm^{\star}$) is easily accessible by the proposed algorithm. Thus, problems due to multiple local optima do not occur on these data sets. Moreover, it indicates the values of the ICL and MICL criteria related to the model selected by MS. Even if the MICL value is always larger (or equal) than the ICL value, this difference is often small. Thus, the partitions $\hat{\tz}_{\bm^{\star}}$ and $\tz^{\star}_{\bm^{\star}}$ are often quite similar.

\begin{table}[ht!]
\begin{center}
\begin{tabular}{ccccccc}
& banknote & coffee &  wine & cancer&   golub\\ 
\hline Occurrence & 48 & 27 & 11  & 48 &  6 \\ 
MICL & -1009.2 & -644.1 & -3715.7 & -7963.5 & -103858.8 \\ 
ICL & -1009.2 & -644.1 & -3715.8 &-8064.1 & -103858.8\\
  \hline
\end{tabular} 
\end{center}
\caption{Ocurrence where the couple ($\tz^{\star}_{\bm^{\star}},\bm^{\star}$) has been found by the algorithm for the 50 random initializations and values of the information criteria for model $\bm^{\star}$. \label{tab::local}}
\end{table}

\subsection{Component number unknown}
Table~\ref{tab::resunknown} presents the results obtained when the number of components is unknown by setting $g_{\max}=6$. 

\begin{table}[ht!]
\begin{center}
\begin{tabular}{c|cc|cc|cc|cc}
\multicolumn{1}{c}{Data} & \multicolumn{2}{c}{$\hat{g}$} & \multicolumn{2}{c}{NRV} & \multicolumn{2}{c}{ARI} & \multicolumn{2}{c}{BIC}  \\ 
& MS & RD& MS & RD& MS & RD& MS & RD \\ 
\hline banknote & 3 & 5 & 6 & 4 & 0.61 & 0.40 & -926 & -978 \\ 
coffee & 2 & 3 & 5 & 6 & 1.00 & 0.38 & -522 & -521 \\ 
wine & 4 & 4 & 11 & 6 & 0.67 & 0.72 & -3502 & -3739 \\ 
cancer & 6 & 6 & 13 & 15 & 0.21 & 0.23 & 4192 & 4861 \\ 
golub & 2 & 6 & 553 & 8 & 0.79 & 0.00 & -90348 & -95098 \\ 
\hline
\end{tabular} 
\end{center}
\caption{Results obtained when the component number is unknown: number of components ($\hat{g}$), number of relevant variables (NRV), adjusted rand index computed on the selected model (ARI) and BIC criterion value.\label{tab::resunknown}}
\end{table}

On these data sets, RD selects more components than MS. Thus, its interpretation becomes more complex even if RD can select less variables. The previous remarks are validated by this application. Indeed, we observe than MS can provide a model with a better BIC value. Moreover, when the number of variables increases, the results of RD are damaged due to the lack of optimality. This is shown particularly by \emph{golub} data set.

\section{Discussion} \label{sec::conclusion}
We have proposed a new information criterion to carry out model selection of a finite mixture model.  This criterion can be used for selecting the relevant variables for model-based clustering
in Gaussian mixture settings assuming conditional independence. In such a case, the criterion has a closed form and the model maximizing it is accessible by an algorithm of alternated optimization. Its originality consists in allowing a model selection procedure which does not require the maximum likelihood estimate. 

The criterion can be easily used when the model at hand is a mixture of distributions belonging to an exponential family. Indeed, in such cases, the closed form is preserved. Thus, the MICL criterion can carry out variable selection in a cluster analysis of categorical or mixed datasets by using the model of \citet{Cel91} and of \citet{Mou05} respectively. The application of the proposed method to the categorical data appears to be especially pertinent since conditional independence assumption is often assumed for such data, and since Jeffreys non informative prior distributions are available.

If the conditional independence assumption is relaxed, the algorithm used for model selection should be modified. Then, its model step is not explicit but it can be achieved by a MCMC method.

We have compared our method with two standard procedures of variable selection in cluster analysis. It was shown that the proposed method outperforms both other ones for the task of variable selection. It results in a better partitioning accuracy. In a moderate computing time, the proposed method can manage datasets with a large number of variables and a relatively large number of individuals. However, the procedure of model selection is time-consuming if a huge number of individuals is observed. In such a case, the optimization of the model selection procedure is an issue which calls for further improvements. 

Finally, this approach could be extend to perform a more elaborated variable selection. Indeed, by using the approach of \citet{Mau09}.

The R package \emph{VarSelLCM} implementing the proposed method is downloadable on CRAN.

\section*{Acknowledgement}
The authors are grateful to Gilles Celeux and Jean-Michel Marin for their leading comments.

\bibliographystyle{apalike}
\bibliography{biblio.bib}

\begin{appendices}
  \section{Consistency of the MICL criterion} \label{sec::consistence}

  This section is devoted to the proof of consistency of our $\MICL$ criterion
  with a fixed number of components. The first part deals with non-nested
  models and requires a \textit{biais-entropy} compensation assumption. The
  second part covers the nested models, \textit{i.e}, when the competing model
  contains the true model. In what follows, we consider the  true model $\bm^{\pzero} =
  \big(g^{\pzero}, \boldsymbol{\omega}^{\pzero}\big)$, its set of relevant
  variables is $\Omega^{\pzero} = \left\{j : \omega^{\pzero}_j = 1 \right\}$
  and the parameter is $\bth^{\pzero}$.

\paragraph{Case of non-nested model}
 We need to introduce the entropy notation given by
\begin{equation*}
  \xi\big(\bth; \tz, \bm\big) = \sum_{i = 1}^n \sum_{k = 1}^g z_{ik} \ln \tau_{ik}\big(\bth \mid \bm\big),
\end{equation*}
where $\tau_{ik}\big(\bth \mid \bm\big) 
       = \dfrac{\pi_k \phi\big(\boldsymbol{x}_i \mid \theta_k,
         \bm\big)}{\sum_h^g\pi_h 
         \phi\big(\boldsymbol{x}_i \mid \theta_h, \bm\big)}.$ 
  \begin{proposition}
    Assume that $\bm^{\pone}$ is a model such that $\bm^{\pzero}$ is a non-nested
    within $\bm^{\pone}$. Assume that  
    \begin{equation}
    \label{eq:H}
        - \mathbb{E}\left[\ln \dfrac{\sum_{k = 1}^{g^{\pzero}}\pi_k
            \prod_{j = 1}^d\phi\big(x_{1j}
          \mid \mu^{\pzero}_{kj}, \sigma^{\pzero 2}_{kj}\big)\mathbbm{1}_{G^{\pzero}_k}\big(\boldsymbol{x}_1\big)}
         {p\big( \boldsymbol{x}_1 \mid
        \bth^{\pzero},\bm^{\pzero}\big)}\right]
      \le \mathbf{KL}\Big[\bm^{\pzero}||\bm^{\pone}\Big],
   \end{equation}
   where $\mathbf{KL}\Big[\bm^{\pzero}||\bm^{\pone}\Big]$ is the
   Kullback-Leibler divergence of $p\big(\cdot \mid
   \bth^{\pzero},\bm^{\pzero}\big)$ from $p\big(\cdot \mid
   \bth^{\pone},\bm^{\pone}\big)$ and 
   \begin{equation*}
     G^{\pzero}_k = \left\{x \in \mathbb R^d
       : k = \underset{1 \le h \le g^{\pzero}}{\argmax}\, \pi_h \prod_{j = 1}^d\phi\big(x_{1j}
          \mid \mu^{\pzero}_{hj}, \sigma^{2 \pzero}_{hj}\big)\right\}.
 \end{equation*}
 When $n \to \infty$, we have
    \begin{equation*}
      \bbP\bigg(\MICL\big(\bm^{\pone}\big) > \MICL\big(\bm^{\pzero}\big)\bigg)
      \longrightarrow 0.
    \end{equation*}
    \end{proposition}

    \begin{proof}
      For any model $\bm$, we have the following inequalities, 
      \begin{equation*}
      \ICL\big(\bm\big) \le \MICL\big(\bm\big) \le \ln p\big(\tx \mid \bm\big).
      \end{equation*}
       It follows,
     \begin{equation*}
      \bbP\bigg\{\MICL\big(\bm^{\pone}\big) - \MICL\big(\bm^{\pzero}\big)
        > 0\bigg\}
       \le \bbP\bigg\{\ln p\big(\tx \mid \bm^{\pone}\big) -
       \ICL\big(\bm^{\pzero}\big) > 0\bigg\}.
    \end{equation*}
     Now set $\Delta\nu = \nu^{\pone} - \nu^{\pzero}$ where $\nu^{\pone}$ and
      $\nu^{\pzero}$ are the numbers of free parameters in the models
      $\bm^{\pone}$ and $\bm^{\pzero}$ respectively. 
     Using Laplace's approximation, we have
     \begin{equation*}
       \ICL\big(\bm^{\pzero}\big) 
         = \ln p\left(\tx \mid \widehat{\bth}^{\pzero}, \bm^{\pzero}\right)
       + \xi\left(\widehat{\bth}^{\pzero}; \widehat{\tz}^{\pzero},
         \bm^{\pzero}\right) - \dfrac{\nu^{\pzero}}{2} \ln n+ \mathcal{O}_p(1),
     \end{equation*}
     where $\widehat{\bth}^{\pzero}$ and $\widehat{\tz}^{\pzero}$ are
     respectively the MLE and the partition given by the
     corresponding MAP rule.
     In the same way, we have 
     \begin{equation*}
       \ln p\big(\tx \mid \bm^{\pone}\big) = \ln p\big(\tx \mid
       \widehat{\bth}^{\pone},
       \bm^{\pone}\big) - \dfrac{\nu^{\pone}}{2} \ln n + \mathcal{O}_p(1),
     \end{equation*}
    where $\widehat{\bth}^{\pone}$ is the MLE of $\bth^{\pone}$.
    Note that 
    \begin{align*}
     \ln p\big(\tx \mid \bm^{\pone}\big) -
       \ICL\big(\bm^{\pzero}\big)  & = \dfrac{A_n}{2} + n B_n  - \dfrac{\Delta
       \nu}{2} \ln n + \mathcal{O}_p(1),    
    \end{align*}
   where
   \begin{equation*}
     A_n  = 2 \ln\dfrac{p\left(\tx \mid \widehat{\bth}^{\pone},\bm^{\pone}\right)}{ 
       p\left(\tx \mid \bth^{\pone}, \bm^{\pone}\right)} - 
      2 \ln \dfrac{p\left(\tx \mid \widehat{\bth}^{\pzero},
          \bm^{\pzero}\right)}{p\left(\tx \mid \bth^{\pzero},
          \bm^{\pzero}\right)},
   \end{equation*}
and 
\begin{equation*}
  B_n = \dfrac{1}{n} \ln\dfrac{p\left(\tx \mid \bth^{\pone},\bm^{\pone}\right)}
   {p\left(\tx \mid \bth^{\pzero},\bm^{\pzero}\right)} -
  \dfrac{1}{n}\xi\left(\widehat{\bth}^{\pzero};
    \widehat{\tz}^{\pzero},\bm^{\pzero}\right).
\end{equation*}
When $n \to \infty$, we have $A_n \to \chi^2_{\Delta \nu}$ in
distribution and $B_n$ tends to 
\begin{equation*}
-\mathbf{KL}\Big[\bm^{\pzero}||\bm^{\pone}\Big] - \mathbb{E}\left[\ln \dfrac{\sum_{k = 1}^{g^{\pzero}}\pi_k
            \prod_{j = 1}^d\phi\big(x_{1j}
          \mid \mu^{\pzero}_{kj}, \sigma^{\pzero 2}_{kj}\big)\mathbbm{1}_{G^{\pzero}_k}\big(\boldsymbol{x}_1\big)}
         {p\big( \boldsymbol{x}_1 \mid
        \bth^{\pzero},\bm^{\pzero}\big)}\right]
\end{equation*}
in probability. Thus, under the assumption~\eqref{eq:H}, $\MICL$ is consistent
since when $n \to \infty$, we have
\begin{align*}
  \bbP\bigg\{\MICL\big(\bm^{\pone}\big) - \MICL\big(\bm^{\pzero}\big)
        >  0 \bigg\} & \le  \bbP\bigg[A_n    + 
  \mathcal{O}_p(1) >  \Delta \nu \ln n\bigg] + \bbP\bigg[B_n > 0 \bigg]\\
  & \longrightarrow 0. 
 \end{align*}
\end{proof}

\paragraph{Case of nested model}
 Recall that
 $\MICL\big(\bm^{\pzero}\big) = \ln p\big(\tx, \tz^{\pzero}\mid \bm^{\pzero}\big)$,
 where  $\tz^{\pzero} = \underset{\tz}{\argmax} 
 \ln p\big(\tx, \tz \mid \bm^{\pzero}\big)$.
We have 
\begin{equation*}
\tz^{\pzero} = \underset{\tz}{\argmax}\Big\{\ln p\left(\tz \mid
   g^{\pzero}\right)  +  \underset{j \in \Omega_0}{\sum}\ln
 p\big(\tx_{\bullet j} \mid
   \omega_j^{\pzero}, g^{\pzero}, \tz\big) \Big\},
 \end{equation*}
where $\Omega_0 = \left\{j : \omega^{\pzero}_j = 1\right\}$. Let $\bm^{\pone}
= \left(g^{\pzero}, \Omega_1\right)$ where $\Omega_1 = \Omega_0 \cup
\Omega_{01}$ and $ \Omega_{01} = \left\{j : \omega_j^{\pone} = 1,
  \omega_j^{\pzero} = 0 \right\}$. Then, in the same way, we have
$\MICL\big(\bm^{\pone}\big) = \ln p\big(\tx, \tz^{\pone}\mid
  \bm^{\pone}\big)$,
where
\begin{equation*}
 \tz^{\pone} = \underset{\tz}{\argmax}\left[\ln p\left(\tz \mid
      g^{\pzero}\right)  +  \underset{j \in \Omega_1}{\sum}\ln
    p\left(\tx_{\bullet j} \mid
      \omega_j^{\pone}, g^{\pzero}, \tz\right) \right].
\end{equation*}
Let $j \in \Omega_{01}$, Laplace's approximation gives us,  
\begin{equation*}
 \ln  p\left(\tx_{\bullet j} \mid \omega_j^{\pone}, g^{\pzero},
   \tz\right)
 = \sum_{i =1}^n \sum_{k = 1}^g z_{ik}\ln 
\phi\left(x_{ij} \mid \tilde{\mu}^{\pone}_{kj}, \tilde{\sigma}^{\pone
        2}_{kj}\right) - g^{\pzero} \ln n + \mathcal{O}_p\big(1\big), 
\end{equation*}
where
\begin{equation*}
  \left(\tilde{\mu}^{\pone}_{kj}, \tilde{\sigma}^{\pone
        2}_{kj}\right) =  \underset{\mu^{\pone}_{kj}, \sigma^{\pone
        2}_{kj}}{\argmax} \sum_{i =1}^n z_{ik}\ln 
   \phi\left(x_{ij} \mid \mu^{\pone}_{kj}, \sigma^{\pone
        2}_{kj}\right).
\end{equation*}
\begin{proposition}
Assume that  $\bm^{\pone}$ is a model such that $g^{\pone} = g^{\pzero}$ and $\Omega_1 =
  \Omega_0 \cup \Omega_{01}$ where $\Omega_{01} \neq \emptyset$, i.e, 
the model $\bm^{\pzero}$ is nested within the model $\bm^{\pone}$ with
the same number of components. When $n \to \infty$,
\begin{equation*}
  \bbP\bigg(\MICL\big(\bm^{\pone}\big) > \MICL\big(\bm^{\pzero}\big)\bigg)
 \longrightarrow 0.
\end{equation*}
\end{proposition}
\begin{proof}
We have 
 \begin{equation*}
   \bbP\bigg\{\MICL\big(\bm^{\pone}\big) > \MICL\big(\bm^{\pzero}\big)\bigg\}
   \le  
   \bbP\left\{\sum_{j \in \Omega_{01}} \ln \dfrac{p\big(\tx_{\bullet j} \mid
   \omega^{\pone}_j, g^{\pzero}, \tz^{(1)}\big)}{
    p\big(\tx_{\bullet j} \mid
   \omega^{\pzero}_j, g^{\pzero}, \tz^{(0)}\big)}
   > 0
   \right\},
   \end{equation*}
And for each $j \in \Omega_{01}$,
when $n \to \infty$
\begin{equation*}
2 \sum_{i =1}^n \sum_{k =1}^{g^{\pzero}} z^{\pone}_{ik}\ln 
\dfrac{\phi\left(x_{ij} \mid \tilde{\mu}^{\pone}_{kj}, \tilde{\sigma}^{\pone
        2}_{kj}\right)}{\phi\left(x_{ij} \mid \mu^{\pzero}_{1j},\sigma^{\pzero2}_{1j}\right)} 
\longrightarrow \chi^2_{2g^{\pzero}}  \quad \text{in distribution}.
\end{equation*}
We have
 \begin{align*}
 \bbP\bigg(\sum_{j \in \Omega_{01}} \ln \dfrac{p\big(\tx_{\bullet j} \mid
   \omega^{\pone}_j, g^{\pzero}, \tz^{(1)}\big)}{
    p\big(\tx_{\bullet j} \mid
   \omega^{\pzero}_j, g^{\pzero}, \tz^{(0)}\big)}
   > 0
   \bigg)  
   & = 
   \bbP\Big(\chi^2_{2(g^{\pzero}-1)} - 2 (g^{\pzero}-1) \ln
   n  > 0\Big)\\
   & \longrightarrow 0 \quad \text{ by Chebyshev's inequality}. 
  \end{align*}
\end{proof}
 \end{appendices}

\end{document}